\documentclass[pra,aps,floatfix,amsmath,amssymb,superscriptaddress,tightenlines]{revtex4}
\usepackage{graphicx}% Include figure files
\usepackage{epstopdf}
\newcommand{\be}{\begin{equation}}
\newcommand{\ee}{\end{equation}}
\usepackage{amsmath}
\usepackage{amsfonts}
\usepackage{bm}
\usepackage{color}
\usepackage{subfigure}
\usepackage{amsthm}
\usepackage[latin1]{inputenc}

\newcommand{\Tr}{\textrm{Tr}}
\newcommand{\defeq}{\mathrel{\mathop:}=}

\newtheorem{lem}{Lemma}
\newtheorem{thm}{Theorem}
\newtheorem{defi}{Definition}
\newtheorem{corollary}{Corollary}

\begin{document}
\title{Modulus of convexity for operator convex functions}
\author{Isaac H. Kim}
\affiliation{Perimeter Institute of Theoretical Physics, Waterloo ON N2L 2Y5, Canada}
\affiliation{Institute of Quantum Information and Matter, Pasadena CA 91125, USA}

\date{\today}
\begin{abstract}
Given an operator convex function $f(x)$, we obtain an operator-valued lower bound for $cf(x) + (1-c)f(y) - f(cx + (1-c)y)$, $c \in [0,1]$. The lower bound is expressed in terms of the matrix Bregman divergence. A similar inequality is shown to be false for functions that are convex but not operator convex.
\end{abstract}

\maketitle

\section{Introduction}

Formally, a function $f(x)$ is {\it operator convex} if it satisfies the following inequality for self-adjoint operators $A$ and $B$:
\begin{equation}
cf(A) + (1-c)f(B) - f(cA + (1-c)B) \geq 0, c\in [0,1]. \label{eq:operator_convexity}
\end{equation}
Most of the interesting examples deal with operators that are positive semi-definite. We shall follow the same convention in this paper.

Operator convex functions are known to satisfy a number of interesting properties. An important discovery was made by Hansen and Pederson, who used Eq.\ref{eq:operator_convexity} in order to obtain an operator generalization of the Jensen inequality.\cite{Hansen1982} Recently, Effros provided an elegant proof of the strong subadditivity of entropy (SSA)\cite{Effros2009} by (i) defining an operator generalization of the perspective function and (ii) using the aforementioned operator generalization of Jensen inequality. These results show that the fundamental inequality in quantum information theory - the strong subadditivity of entropy\cite{Lieb1972} - can be essentially derived from the operator convexity of a certain matrix-valued function.

An important open question in quantum information theory concerns the structure of states that are approximately conditionally independent, {\it i.e.,} the structure of states that  has a small yet nonzero quantum conditional mutual information.\footnote{Strong subadditivity asserts that quantum conditional mutual information is nonnegative.} The motivation comes from the fact that states that satisfy the equality condition of the SSA forms a quantum Markov chain.\cite{Hayden2004} One natural speculation along this line is to guess that a quantum state with a small conditional mutual information is close to some quantum Markov chain state. While this intuition is correct for classical states, its obvious quantum generalization is known to be false.\cite{Ibinson2007}

There are several ways to circumvent this issue. The predominant approach in the literature is to replace the set of quantum Markov chain states to a larger set of states, namely the separable states. The first result in this direction was obtained by Brand\~ao {\it et al.}\cite{Brandao2010} Their result was subsequently strengthened by Li and Winter.\cite{Li2012}

However, there is another possibility that is not necessarily precluded by the counterexamples of Ibinson {\it et al.}\cite{Ibinson2007} The quantum Markov chain property derived in Ref.\cite{Hayden2004} is a consequence of Petz's theorem\cite{Petz1988} and Koashi-Imoto theorem\cite{Koashi2002}. Petz's theorem asserts that, if a relative entropy between two quantum states does not decrease under a quantum channel, there exists a canonical recovery operation that can perfectly reverse the action of the channel. Koashi-Imoto theorem concerns the structure of states which are invariant under certain quantum channels. The quantum Markov chain property is derived in Ref.\cite{Hayden2004} by applying Koashi-Imoto theorem to Petz's canonical recovery operation.

Therefore, one may consider an alternative possibility: there might exist a canonical recovery operation analogous to Petz's recovery channel, whose performance is determined by the conditional mutual information. Such a result may not contradict the counterexamples of Ibinson {\it et al.}, since one cannot directly apply Koashi-Imoto theorem to such channels when the recovery operation is not perfect. Since Petz's theorem is based on the fundamental results about operator convex and operator monotone functions, a strengthening of these results may lead to new insights on the structure of states that have a small but nonzero amount of conditional mutual information. Also, obtaining such a strengthening might be interesting in its own right; it might be potentially useful in extending the preexisting results that are based on the properties of operator convex functions.

Motivated from these observations, we obtain a possible strengthening of Eq.\ref{eq:operator_convexity}. More precisely, the right hand side of Eq.\ref{eq:operator_convexity} shall be replaced by an operator-valued function that is always nonnegative. This operator, up to some constant that depends on $c$, is the {\it matrix Bregman divergence}.\cite{Petz2007} Matrix Bregman divergence is a natural matrix generalization of the classical Bregman divergence.\cite{Bregman1967} Given a convex function $f(x)$ and two probability distributions $p(x), q(x)$ over the same domain, Bregman divergence can be defined as
\begin{equation}
B_{f}(p \| q) := f(p) - f(q) - \lim_{c \to 0+} \frac{f(q+(p-q)c) - f(q)}{c}. \nonumber
\end{equation}
Petz noted that a similar treatment can be carried out even when $p$ and $q$ are promoted to operators, provided that the function $f$ is operator convex. The resulting matrix-valued divergence is the matrix Bregman divergence. Interestingly, the strengthening is only applicable to operator convex functions; the inequality is false for a convex function that is not operator convex.

As an application of this result, we prove an inequality that extends Pinsker's inequality. Recall that Pinsker's inequality asserts that the relative entropy $D(\rho \| \sigma) := \Tr(\rho (\log \rho - \log \sigma))$ between two normalized positive semidefinite operators, $\rho$ and $\sigma$, is lower bounded by their trace distance:
\begin{equation}
D(\rho \| \sigma) \geq \frac{1}{2}\|\rho - \sigma\|_1^2. \nonumber
\end{equation}
A simple corollary of our main result is the following inequality:
\begin{equation}
S(c\rho + (1-c)\sigma) - cS(\rho) - (1-c)S(\sigma) \geq \frac{1}{2}c(1-c)\|\rho- \sigma\|_1^2, \quad c\in [0,1], \nonumber
\end{equation}
where the underlying Hilbert space is finite-dimensional. 

The rest of the paper starts by describing the main result in Section \ref{section:convexity_BregmanDivergence}. We shall also show that the main result cannot be generalized to convex functions by providing a simple argument. Section \ref{section:AH_strengthening} describes the key technical result of this paper, which is a strengthening of the well-known Arithmetic-Harmonic inequality. Using this strengthening, we prove the main result in Section \ref{section:proof}.

\section{An inequality between the modulus of convexity and Bregman divergence for operator convex functions\label{section:convexity_BregmanDivergence}}
In order to describe the main result, we set the notations first. 
\begin{defi}
Modulus of convexity of a function $f(x)$ is
\begin{equation}
\mathcal{C}_f^{c}(A,B) \defeq c f(A) + (1-c)f(B) -f(c A + (1-c)B), \quad c\in [0,1]. \nonumber
\end{equation}
\end{defi}
Note that the modulus of convexity of an operator convex function is always nonnegative, as long as $A$ and $B$ are self-adjoint operators whose spectrum lie on the domain of $f$. The matrix Bregman divergence is nonnegative due to the same reason. Following Petz\cite{Petz2007}, we define the matrix Bregman divergence as follows.
\begin{defi}
Bregman divergence $\mathcal{D}_f(A,B)$ is
\begin{equation}
\mathcal{D}_f(A,B) = f(A) -f(B) - \lim_{t \to 0+} t^{-1}(f(B+t(A-B))- f(B)). \nonumber
\end{equation}
\end{defi}

The main result asserts that the Bregman divergence provides an operator-valued lower bound for the modulus of convexity.
\begin{thm}
For $A,B>0$, if $f(x)$ on $[0, \infty)$ is operator convex and $0<c < 1$
\begin{align}
\mathcal{C}_f^{c}(A,B) \geq &c(1-c) \frac{\mathcal{D}_f(M(1-c),M(c))}{(1-2c)^2 } &\quad c \neq \frac{1}{2} \nonumber \\
&  \frac{1}{8} \frac{d^2}{dx^2} f(M(\frac{1}{2}+x)) |_{x=0}. &\quad c = \frac{1}{2}, \label{eq:main_result}
\end{align}
where $M(c):=cA + (1-c)B$.
\label{thm:mainresult}
\end{thm}

\subsection{Convex vs. operator convex functions}
A natural question is whether Theorem \ref{thm:mainresult} can be extended to operator convex functions of order $n$. We will give a simple argument that such an extension cannot exist for $n=1$. Recall that operator convex functions of order $1$ refer to all the convex functions. One can easily check that the function $g(x) = \frac{1}{2}x^2 - (1+ x) \log (1+x)$ is convex for $x>0$.

Our claim is that Eq.\ref{eq:main_result} cannot hold for such $g(x)$. If Eq.\ref{eq:main_result} holds for $f(x)=g(x)$, it implies that Eq.\ref{eq:main_result} holds for $f(x)= - (1+x)\log (1+x)$ as well; this follows from a simple observation that Eq.\ref{eq:main_result} holds  with an equality if $f(x)=x^2$. Since $f(x)= - (1+x) \log (1+x)$ as well as $f(x)=(1+x)\log(1+x)$ satisfies Eq.\ref{eq:main_result}, one can conclude that the inequalities in Eq.\ref{eq:main_result} must be satisfied with an equality for such function. Clearly, this is not the case, and we arrive at a contradiction. Therefore, Eq.\ref{eq:main_result} cannot be extended to operator convex functions of order $1$.

\section{Strengthening of the Arithmetic-Harmonic inequality\label{section:AH_strengthening}}

In this section, we prove a strengthening of the well-known Arithmetic-Harmonic(AH) inequality. AH inequality states that
\begin{equation}
\frac{1}{A} + \frac{1}{B} \geq \frac{4}{A+B} \label{eq:AH}
\end{equation}
for positive definite matrices $A$ and $B$. It is well known that any operator convex function has a unique integral representation that can utilize Eq.\ref{eq:AH}. For example, the following theorem was recently proved by Hiai {\it et al.}\cite{Hiai2010}
\begin{thm}
\cite{Hiai2010} A continuous real function $f$ on $[0,\infty)$ is operator convex iff there exists a real number $a$, a nonnegative number $b$, and a nonnegative measure $\mu$ on $[0,\infty)$, satisfying
\begin{equation}
\int_0^{\infty} \frac{1}{(1+\lambda)^2}d\mu(\lambda) < \infty  \nonumber,
\end{equation}
such that
\begin{equation}
f(x) = f(0)+ ax + bx^2 + \int_0^{\infty}(\frac{x}{1+\lambda} - 1+ \frac{\lambda}{x+\lambda}) d\mu(\lambda), x\in [0,\infty).
\end{equation}
Moreover, the numbers $a,b$, and the measure $\mu$ is uniquely determined by $f$.
\end{thm}

The existence of the canonical form for operator convex functions is the main motivation behind the strengthening of AH inequality. Our key lemma is the following:
\begin{lem}
For $A,B>0$,
\begin{equation}
\frac{1}{2}(\frac{1}{A} + \frac{1}{B}) - \frac{2}{A+B} \geq 2 \frac{1}{A+B}(A-B)\frac{1}{A+B}(A-B)\frac{1}{A+B}\label{eq:StrongerAH}
\end{equation}
\end{lem}
\begin{proof}
Define $C=A^{-\frac{1}{2}}BA^{-\frac{1}{2}}$. Applying a left and right multiplication of $A^{\frac{1}{2}}$ on both sides of Eq.\ref{eq:StrongerAH}, the left hand side can be expressed as $\frac{(1-C)^2}{2C(1+C)}$, while the right hand side can be expressed as $2\frac{(1-C)^2}{(1+C)^3}$. Using the fact that $(1+C)^2 \geq 4C$, one can establish the inequality. 
\end{proof}

\section{Proof of the main result\label{section:proof}}
A well-known approach for proving operator Jensen inequality involves (i) proving the inequality at the midpoint and (ii) making a judicious choice of matrices in an enlarged Hilbert space.\cite{Hansen1982} We shall follow a similar approach. Under elementary manipulations, one can show that Theorem \ref{thm:mainresult} for a general operator convex function follows by proving it for a special family of functions, namely $f_{\lambda}(x) = \frac{1}{\lambda+x}$. Without loss of generality, we shall prove Theorem \ref{thm:mainresult} for $f(x)=\frac{1}{\lambda+x}$. The case for the linear and quadratic terms are trivial, so we omit the proof for them.

First, we consider the $c=\frac{1}{2}$ case of Theorem \ref{thm:mainresult}.
\begin{align}
\frac{f(A) + f(B)}{2} - f(\frac{A+B}{2})  &=\int^{\infty}_{0} (\frac{1}{2}(\frac{1}{A+\lambda} + \frac{1}{B+\lambda}) - \frac{2}{A+B + 2\lambda}) d\mu(\lambda) \nonumber\\
& \geq 2\int^{\infty}_{0} \frac{1}{A+B + 2\lambda}(A-B) \frac{1}{A+B + 2\lambda}(A-B)\frac{1}{A+B + 2\lambda} d\mu(\lambda) \nonumber \\
&=\frac{1}{4} \int^{\infty}_{0}  \frac{1}{\frac{A+B}{2} + \lambda}(A-B) \frac{1}{\frac{A+B}{2} + \lambda}(A-B)\frac{1}{\frac{A+B}{2} + \lambda} d\mu(\lambda) \nonumber \\
&= \frac{1}{8} \int^{\infty}_{0} \frac{d^2}{dx^2} \frac{1}{\frac{A+B}{2} + x(A-B) + \lambda} |_{x=0} d\mu(\lambda)  \nonumber \\
&= \frac{1}{8} \frac{d^2}{dx^2}f(\frac{A+B}{2} + x(A-B))|_{x=0}
\end{align}

Away from the midpoint, we use the following choice of operators:\cite{Bhatia1997}
\begin{equation}
W=
\begin{pmatrix}
 c^{\frac{1}{2}}I & -(1-c)^{\frac{1}{2}} \\
 (1-c)^{\frac{1}{2}}I & c^{\frac{1}{2}}I
\end{pmatrix}, \nonumber
\end{equation}
and
\begin{equation}
T=
\begin{pmatrix}
A & 0 \\
0 & B
\end{pmatrix}. \nonumber
\end{equation}
Each of the entries in the matrices correspond to a block of square matrices of the dimension. Setting $T_1 = WTW^{\dagger}$, $T_2= W^{\dagger}TW$, and applying it to the midpoint convexity result, one can obtain the desired result. More precisely, the convexity at the midpoint is the following:
\begin{align}
\frac{f(T_1) + f(T_2)}{2} - f(\frac{T_1+T_2}{2}) =
\begin{pmatrix}
c f(A) + (1-c) f(B)  &0 \\
0 & c f(B) + (1-c)f(A)
\end{pmatrix}
-
f
\begin{pmatrix}
c A + (1-c) B & 0 \\
0 & c B + (1-c) A
\end{pmatrix}. \nonumber
\end{align}
One can also easily check the following facts:
\begin{align}
\frac{T_1+ T_2}{2} =
\begin{pmatrix}
c A + (1-c) B  &0 \\
0 & c B + (1-c)A
\end{pmatrix}. \nonumber
\end{align}
\begin{align}
\frac{T_1-T_2}{2} =
\begin{pmatrix}
0 & \sqrt{c(1-c)}(A-B) \\
\sqrt{c(1-c)}(A-B) & 0
\end{pmatrix}. \nonumber
\end{align}
By taking one of the blocks,
\begin{align}
\mathcal{C}_f^{c}(A,B) &\geq c(1-c) \int^{\infty}_{0}  \frac{1}{M(c) + \lambda} (A-B) \frac{1}{M(1-c) + \lambda} (A-B) \frac{1}{M(c) + \lambda} d\mu(\lambda) \nonumber \\
&= c(1-c) \frac{1}{(2c-1)^2} \int^{\infty}_{0}  \frac{1}{M(c) + \lambda} (M(c)-M(1-c)) \frac{1}{M(1-c) + \lambda} (M(c)-M(1-c)) \frac{1}{M(c) + \lambda} d\mu(\lambda).
\end{align}
One can check that
\begin{equation}
\mathcal{D}_f(A,B) = \int^{\infty}_{0}  \frac{1}{B+\lambda} (A-B) \frac{1}{A+\lambda} (A-B) \frac{1}{B+\lambda} d\mu(\lambda),
\end{equation}
completing the proof.

\section{Application to the von Neumann entropy}
Now we discuss an application of Theorem \ref{thm:mainresult}.
\begin{corollary}
For density matrices $\rho, \sigma$ on a finite-dimensional Hilbert space,
\begin{equation}
S(c\rho  + (1-c)\sigma) - cS(\rho) - (1-c)S(\sigma) \geq \frac{1}{2}c(1-c) \|\rho - \sigma\|_1^2
\end{equation}
\end{corollary}
\begin{proof}
For $f(x) = x\log x$, Petz showed that
\begin{equation}
\Tr(\mathcal{D}_f(A,B)) = D(A\|B),
\end{equation}
where $D(A \|B) = \Tr(A (\log A - \log B))$ is the relative entropy between $A$ and $B$.\cite{Petz2007} Hence, the following inequality immediately follows.
\begin{equation}
S(c\rho  + (1-c)\sigma) - cS(\rho) - (1-c)S(\sigma) \geq c(1-c) \frac{1}{(1-2c)^2}D( c\sigma + (1-c)\rho\|  c\rho + (1-c)\sigma)
\end{equation}
for $c\neq \frac{1}{2}$. Applying Pinsker's inequality,
\begin{equation}
S(c\rho  + (1-c)\sigma) - cS(\rho) - (1-c)S(\sigma) \geq \frac{1}{2}c(1-c) \|\rho - \sigma\|_1^2. \label{eq:strong_concavity_entropy}
\end{equation}
Eq.\ref{eq:strong_concavity_entropy} should be true for $c=\frac{1}{2}$ as well by some continuity argument, which is discussed below.

Recall that Fannes' inequality asserts that
\begin{equation}
|S(\rho) - S(\sigma)| \leq \epsilon \log d - \epsilon \log \epsilon,
\end{equation}
where $\epsilon = \|\rho - \sigma\|_1$ and $d$ is the dimension of the Hilbert space.\cite{Fannes1973} Define $\mathcal{S}(c)$ as
\begin{equation}
\mathcal{S}(c) = S(c\rho + (1-c)\sigma) - cS(\rho) - (1-c)S(\sigma). \nonumber
\end{equation}
Using Fannes' inequality,
\begin{equation}
|\mathcal{S}(\frac{1}{2}) - \mathcal{S}(\frac{1}{2} + \delta)| \leq \epsilon \delta(\log d - \log \epsilon \delta)+ 2\delta \log d. \label{eq:convexity_difference}
\end{equation}
Denoting the right hand side of Eq.\ref{eq:convexity_difference} as $\Delta(\delta, \epsilon, d)\geq 0$,
\begin{align}
\mathcal{S}(\frac{1}{2}) & \geq \mathcal{S}(\frac{1}{2}+\delta) - \Delta(\delta, \epsilon, d) \nonumber\\
& \geq \frac{1}{8} \|\rho-\sigma\|_1^2 - \frac{1}{2}\delta^2\|\rho-\sigma\|_1^2 - \Delta(\delta, \epsilon, d). \nonumber
\end{align}
Taking the $\delta \to 0$ limit, we obtain:
\begin{equation}
\mathcal{S}(\frac{1}{2}) \geq \frac{1}{8}\|\rho - \sigma\|_1^2 \nonumber
\end{equation}
\end{proof}

\section{Discussion}
We have obtained a lower bound for the modulus of convexity for operator convex functions, which can be expressed in terms of the matrix Bregman divergence. We also gave a simple argument that the inequality cannot be extended to general convex functions. For the operator convex function $f(x)= x\log x$, the trace of the matrix Bregman divergence reduces to quantum relative entropy. In this case, the inequality reduces to the \emph{strict concavity} of Von Neumann entropy. It will be interesting to find an application of this inequality. Another important question is to find a strengthening of the operator Jensen inequality. Since many of the nontrivial results in quantum information theory can be essentially derived from the operator Jensen inequality, its strengthening will be undoubtedly useful in many contexts.

\begin{acknowledgements}
I would like to thank Andreas Winter and Alexei Kitaev for many helpful discussions which motivated this work. I would also like to thank Jon Tyson, Mary Beth Ruskai, Fernando Brand\~ao for helpful discussions. I would also like to thank Lin Zhang for pointing out an error in the original manuscript. Lastly, I thank the anonymous referee who suggested to investigate whether the main result holds for operator convex function of finite order. This research was supported in part by NSF under Grant No. PHY-0803371, by ARO Grant No. W911NF-09-1-0442, and DOE Grant No. DE-FG03-92-ER40701. Research at Perimeter Institute is supported by the Government of Canada through
Industry Canada and by the Province of Ontario through the Ministry of Economic Development and Innovation.
\end{acknowledgements}

\end{document}